\documentclass[a4paper,11pt] {article}

\usepackage[english]{babel}
\usepackage{graphicx}
\usepackage{enumerate}
\usepackage{amsmath}
\usepackage{color, colortbl}

\newcommand{\vsd}{\vspace{+.2cm} }

\newtheorem{theorem}{Theorem}[section]
\newtheorem{defin}[theorem]{Definition}
\newtheorem{lemma}[theorem]{Lemma}

\newtheorem{definition}[theorem]{Definition}
\newtheorem{cor}[theorem]{Corollary}
\newtheorem{rqe}[theorem]{Remark}
\newtheorem{prop}[theorem]{Proposition}
\newtheorem{proposition}[theorem]{Proposition}

\newenvironment{proof}{\par\noindent {\bf Proof.} \rm}{\ ~~~$\fbox{}$}
\newenvironment{proof2}[1]{\par\noindent {\bf Proof of #1.} \rm}{\ ~~~$\fbox{}$}

\newcommand{\Z}{\mbox{\rm \lower0.3pt\hbox{$\angle\!\!\!$}Z}}

\newcommand{\sub}[2]{#1{[#2]}}

\newcommand{\card}[1]{{{\rm Card(}#1{\rm)}}}

\newcommand{\neww}[1]{\omega}
\newcommand{\newx}[1]{\chi}
\newcommand{\newy}[1]{\gamma}

\newcommand{\Modific}[1]{#1'}

\newcommand{\factors}[1]{{\rm Fcts} \left( #1 \right)}

\newcommand{\fromAtoB}{from $A^*$ to $B^*$}

\title{An overlap-free morphism is a $k$-power-free morphism\\
for any integer $k \geq 3$}
\author{Francis Wlazinski}
\date{\today}

\begin{document}

\maketitle

\abstract{It's all in the title.}

\parindent=0cm
\parskip=0.15cm

\newcounter{comptnivun}
\setcounter{comptnivun}{1}
\newcounter{comptnivdeux}
\setcounter{comptnivdeux}{1}
\newcounter{comptnivtrois}
\setcounter{comptnivtrois}{1}
\newcounter{comptnivquatre}
\setcounter{comptnivquatre}{1}

\bibliographystyle{plain}


\section{\label{sectionPreliminaries} 
         Introduction and preliminaries}

Let us recall some basic notions of Combinatorics of words we will use in this paper.

\subsection{\label{sectionWords}Words}

In the following,  $A$ and $B$ are \textit{alphabets}, that is, finite sets of symbols called \textit{letters}.
Since an alphabet with one element is of limited interest to us, 
we always assume that the cardinality of alphabets is at least two.

A \textit{word} is an element in $A^*$, the free monoid generated by $A$, whose identity element is the empty word,
denoted $\varepsilon$, and whose composition law, usually  unnoted, is simply the juxtaposition of symbols. 
We denote by $A^+$ the set of non-empty words, that is, $A^+=A^* \setminus \{\varepsilon\}$. 
We also speak of product of words, just as we speak of juxtaposition of words.  
We will not always specify the alphabet used, as this will often be irrelevant.

Given a non-empty
word $u = a_1\ldots a_n$, with $a_i \in A$ for every integer $i$ from 1 to $n$, the \textit{length}
of $u$ denoted by $|u|$ is the integer $n$,
that is, the number of letters of $u$.
By convention, we have $|\varepsilon| = 0$.
The \textit{mirror image} of $u$, denoted by $\tilde{u}$, is the word 
$a_{n} \ldots a_{2}a_{1}$. 

\vspace{+0.07cm}

A word $u$ is a \textit{factor} of a word $v$ if there exist
two (possibly empty) words $p$ and $s$ such that $v = p u s$.
We denote by ${\rm Fcts}(v)$ the set of all factors of $v$.
If $u\in{\rm Fcts}(v)$, we also say that $v$ \textit{contains} the word $u$ (as a factor).
If $p = \varepsilon$, $u$ is a \textit{prefix} of $v$.
If $s = \varepsilon$, $u$ is a \textit{suffix} of $v$.
If $u \neq v$, $u$ is a \textit{proper} factor of $v$.
If $u$, $p$, and $s$ are non-empty words, $u$ is an \textit{internal} factor of $v$.

\vspace{+0.07cm}

Two non-empty words $u$ and $v$ are \textit{conjugated} if $u=t_1t_2$ and $v=t_2t_1$ for two 
(possibly empty) words $t_1$ and $t_2$.
If $t_1 \neq \varepsilon$ and $t_2 \neq \varepsilon$, we say 
that $v$ is a \textit{proper} conjugated word of $u$.

\vspace{+0.07cm}

Let $w$ be a non-empty word and let $i, j$ be two integers such that
$0 \leq i-1 \leq j \leq |w|$.
We denote by $\sub{w}{i..j}$ the factor of $w$ such that $|\sub{w}{i..j}|=j-i+1$ 
and $w = p \sub{w}{i..j} s$ for two words $s$ and $p$ satisfying $|p| = i-1$.
Note that, when $j = i - 1$, we have $\sub{w}{i..j} = \varepsilon$.
When $i=j$, we also denote by $\sub{w}{i}$ the factor $\sub{w}{i..i}$, which is the
$i^{\rm th}$ letter of $w$.
In particular, $\sub{w}{1}$ and $\sub{w}{|w|}$ are respectively the first and the last
letter of $w$.

\vspace{+0.07cm}

An overlap is a word of the form $xuxux$ where $x \in A$ and $u\in A^*$. Note that
an equivalent definition is obtained by taking $x \in A^+$. An overlap-free word
is a word in which none of the factors are an overlap.

An overlap is said to be \textit{pure} if all its proper factors are overlap-free. 
Let us remark that, if $avava$ is a  pure overlap, then $a\in A$.

Repeating the reasoning that an overlap is either pure or contains an overlap, 
we obtain that a word that contains
an overlap also contains a pure overlap (as a factor of the first).

Powers of a word are defined inductively by $u^0 = \varepsilon$, and 
for every integer $n \geq 1$, $u^n = u u^{n-1}$.
Given an integer $k \geq 2$, since the case $\varepsilon^k$ is of little interest, 
we call a $k$-power any word $u^k$ with $u\neq \varepsilon$.
A $2$-power (resp. a $3$-power) is also called a \textit{square} (resp. a \textit{cube}).
Given an integer $k \geq 2$, a word is \textit{$k$-power-free} if it does not contain any 
$k$-power as factor. A \textit{primitive} word is a word that is not a $k$-power of 
another word whatever the integer $k \geq 2$.
A (non-empty) $k$-power $v^k$ is called \textit{pure} if 
any proper factor of $v^k$ is $k$-power-free.
In particular, we say that $v^k$ is a pure $k$-power of a word $w$ if $v^k\in{\rm Fcts}(w)$ and $v^k$ is pure.
As for pure overlap, repeating the fact that a non-pure $k$-power contains a $k$-power, which is itself pure
or not, 
we obtain that any $k$-power contains a pure $k$-power.
Moreover, if $v^k$ is a pure $k$-power then $v$ is primitive but the converse does not hold.

\begin{rqe}\label{RemPrim}

Every conjugated word of a primitive word is primitive.

\end{rqe}

\begin{rqe}\label{RemMirOver}
A word $u$ is an overlap (resp. a pure-overlap, a $k$-power or a pure $k$-power)
if and only if $\tilde{u}$ is an overlap (resp. a pure-overlap, a $k$-power or a pure $k$-power).

\end{rqe}

\begin{rqe}\label{RemFacUk}
For any non-empty word $u$ and any integer $k\geq 3$, a factor of $u^k$ of length $|u|$ is a conjugated
word of $u$.
As a consequence, any factor of $u^k$ of length greater than $2|u|$ contains an overlap.

\end{rqe}

\begin{rqe}\label{stdkp}
A word cannot start with two different pure $k$-powers or with two different pure overlaps.
\end{rqe}

\begin{rqe}\label{RemCubeOverlap}
For any word $w$ of length less than or equal to $4$, we have the equivalence
between $w$ cube-free and $w$ overlap-free.
\end{rqe}

The following proposition gives the well-known solutions
(see~\cite{Lot1983}) to two elementary equations on words and will be
widely used in the following sections:

\begin{prop}
\label{Lothaire} 
Let $A$ be an alphabet and $u,v,w$ three words
over $A$.
\begin{enumerate}
\topsep0cm
\itemsep0cm
\item \label{Lotcase1}
If $vu=uw$ and $v \neq \varepsilon$, then there exist two words $r$ and
$s$ in $A^*$, and an integer $n$ such that $u=r(sr)^{n}$, $v=rs$
and $w=sr$.

\item \label{Lotcase2}
If $vu=uv$, then there exist a word $w$ in $A^*$, and two integers
$n$ and $p$ such that $u=w^{n}$ and $v=w^{p}$.

\end{enumerate}
\end{prop}

\begin{rqe}\label{RemLotC1}

In Case~\ref{Lotcase1} of Proposition~\ref{Lothaire}, if $|u|>|v|$ then we get that the two words $vu$ 
and $uw$ contain an overlap.

\end{rqe}

We also need a property on words that is an immediate consequence
of Proposition~\ref{Lothaire}(\ref{Lotcase2}).

\begin{lemma}{\rm \cite{Ker1986,Lec1985}}
\label{factint}
If a non-empty word $v$ is an internal factor of $vv$,
i.e., if there exist two non-empty words $x$ and $y$ such that $vv=xvy$,
then there exist a non-empty word $t$ and two integers $i,j \geq 1$
such that $x=t^i$, $y=t^j$, and $v=t^{i+j}$.
\end{lemma}

\begin{rqe}\label{RemFactInt}

Lemma~\ref{factint} means that, if $v$ is primitive then $v$ can not be an internal factor of $vv$.

\end{rqe}

We also use a well-known result on combinatorics on words:

\begin{prop}[Fine and Wilf]{\rm \cite{Lot1983,Lot2002}}
\label{Fine} Let $x$ and $y$ be two words. If a power of $x$ and a
power of $y$ have a common prefix of length at least equal to
$|x|+|y|-gcd(|x|,|y|)$ then $x$ and $y$ are powers of the same
word.
\end{prop}

As a consequence of Proposition~\ref{Fine}, we get:

\begin{cor}[Ker\"anen]{\rm \cite{Ker1986}}
\label{Kera} Let $x$ and $y$ be two words. If a power of $x$ and a
power of $y$ have a common factor of length at least equal to
$|x|+|y|-gcd(|x|,|y|)$ then there exist two words $t_1$ and $t_2$
such that $x$ is a power of $t_1t_2$ and $y$ is a power of
$t_2t_1$ with $t_1t_2$ and $t_2t_1$ primitive words. Furthermore,
if $|x|>|y|$ then $x$ is not primitive.

\end{cor}

\begin{rqe}\label{RemPrimFactInt}

When $t_1t_2$ is a primitive word, the equation 
$t_2t_1t_2t_1=pt_1t_2s$ implies $p=t_2$ and $s=t_1$. Indeed, if $|p|<|t_2|$ or if $|s|<|t_1|$
then $t_1t_2$ is an internal factor of $(t_1t_2)^2$:~a contradiction with Remark~\ref{RemFactInt}.

By induction, we get that, for any integer $\ell \geq 2$, the equation $(t_2t_1)^\ell =p(t_1t_2)^{\ell -1}s$ 
also implies $p=t_2$ and $s=t_1$.

\end{rqe}


\begin{lemma}\label{LemConjugPureOverlap}

Let $a$ and $b$ be two letters of $A$ and let $u$ and $v$ be two words in $A^*$ such that
$au$ and $bv$ are conjugated words.
If $auaua$ is a pure overlap then the same holds for $bvbvb$.

\end{lemma}

\begin{proof}

Let us assume that $au=v_1v_2$ and $bv=v_2v_1$ for two words $v_1$ and $v_2$ in $A^*$.
And, by contradiction, let us assume that $bvbvb$ is not a pure overlap, that is, 
$v_2v_1v_2v_1b=T_1 \, ctctc \, T_2$ for some letter $c$ and some words $t$, $T_1$ and
$T_2$ with $T_1T_2 \neq \varepsilon$. It means that $au \neq bv$ and so $v_1,v_2 \neq \varepsilon$.
Therefore, $a$ is the first letter of $v_1$ and $b$ is the first letter of $v_2$.

If $|T_1| \geq |v_2|$ then $ctctc$ is a factor of $v_1v_2v_1b$, that is, of $auau$: a contradiction
with the hypotheses.
In a same way, if $|T_2| \geq |v_1|$ then $ctctc$ is a factor of $v_2v_1v_2a$: a contradiction.

From now, we assume that $|T_1| < |v_2|$ and $|T_2| <|v_1|$. 

From the equality $v_2v_1v_2v_1b=T_1 \, ctctc \, T_2$
with $T_1T_2 \neq \varepsilon$, we get that $|ct|< |v_2v_1|$.

If $T_2 = \varepsilon$ then $T_1 \neq \varepsilon$, $c=b$,  and 
there exists a non-empty word $X$ such that $v_2=T_1X$ 
and $ctct=Xv_1 \, T_1Xv_1$. It implies that there exists two non-empty words $t_1$ and $t_2$
of the same length such that $T_1=t_1t_2$,
$ct=(Xv_1)t_1=t_2(Xv_1)$.
From Case~\ref{Lotcase1} of the proposition~\ref{Lothaire}, 
there exist two words $r$ et $s$
and an integer $n$ such that $t_2=rs$, $Xv_1=r(sr)^n$ and $t_1=sr$.
If $r \neq \varepsilon$ then $r$ (prefix of $Xv_1)$ starts with $c (=b)$.
And $v_2v_1=t_1t_2Xv_1$, factor of $auau$, contains $rrsr$ where $sr=t_1$ (prefix of $v_2$) 
also starts with $c$.
That is, $rrsr$ contains an overlap: a contradiction.
If $r = \varepsilon$ then $Xv_1=s^n(\neq \varepsilon)$ and $t_1=t_2=s$.
It follows that $v_2v_1=t_1t_2Xv_1$, again, factor of $auau$, contains $s^3$: a contradiction.

If $T_1 = \varepsilon$, 
in the same way as in the case $T_2 = \varepsilon$, and using mirror image,
we obtain a contradiction with the hypotheses.

If $T_1 \neq  \varepsilon$ and $T_2 \neq  \varepsilon$, let $T_2'$ be the word such that
$T_2=T_2'b$.
There exist two non-empty words 
$X$ and $Y$ such that $v_1=YT_2'$, $v_2=T_1X$ 
and $ctctc=Xv_1 \, v_2Y=XYT_2'\, T_1XY$.
Since $|ct| \leq |XY|$,
let $u_1$ and $u_2$ be the words of the same length such that $T_2' T_1=u_1cu_2$, 
$ct=XYu_1$ and $tc=u_2XY$.
Since $ctc=(cu_2)(XY)=(XY)(u_1c)$,
from Case~\ref{Lotcase1} of Proposition~\ref{Lothaire}, 
there exist two words $r$ et $s$
and an integer $n$ such that $cu_2=rs$, $XY=r(sr)^n$ and $u_1c=sr$.
If $|r| = 0$, then $s=u_1c=cu_2$ and $XY=s^n$ with $n \geq 1$.
One of the words $T_1ctc(=T_1s^{n+1})$ or $ctcT_2'(=s^{n+1}T_2')$ is factor of $v_2v_1$ that is of $(au)^2$.  
But the last letter of $T_1$ is the last letter of $cu_2$ that is of $s$
and the first letter of $T_2'$ is the first letter of $u_1c$ that is of $s$.
It means that $auau=(v_1v_2)^2$ contains an overlap: a contradiction.
If $|r| = 1$, then $r=c$ and $s=u_1=u_2$. Since $X$ and $Y$ are non-empty, $Y$ ends with $c$
and $X$ starts with $c$.
It implies that $v_1v_2=YT_2'T_1X$, factor of $au$ contains $cscsc$: a contradiction.
If $|r| \geq 2$, there exists a word $r'$ such that
$r=cr'c$. Therefore, $u_1cu_2$ factor of $v_2v_1$ contains $cr'cr'c$: a final contradiction.
%
%
%
\end{proof}

%

\begin{lemma}\label{LemPurekmu}

Let $a$ be a letter of $A$, let $u$ be a word in $A^*$ and let $k \geq 3$ be an integer.
If $auaua$ is a pure overlap then $(au)^{k-1}$ is $k$-power-free.

\end{lemma}

\begin{proof}

The result is obvious when $k=3$. 
By contradiction, assume that $(au)^{k-1}$  contains a $k$-power $v^k$ with $k \geq 4$.
In particular, we have $0<|v|<|au|$.

If $|v^{k-1}| <|au|$ and, in particular, $|v^{3}| <|au|$, then there exist a conjugated word of $au$ 
(so a factor of $auau$) that contains $v^{3}$:~a contradiction with the fact that $auaua$ is pure.

If $|v^k| \geq |v| +|au|$, then, by Lemma~\ref{Kera} and Remark~\ref{RemPrim},
the word $au$ is not primitive. This means that $auau$ contains a cube: a contradiction with the fact
that $auaua$ is pure.
\end{proof}



\subsection{\label{sectionMorphisms}Morphisms}

A \textit{morphism} $f$ from $A^*$ to $B^*$ is a mapping
from $A^*$ to $B^*$ such that $f(uv) = f(u)f(v)$ for all words $u, v$ in $A^*$.
When $B$ has no importance, we say that $f$ is a morphism on
$A^*$ or that $f$ is defined on $A^*$. 
Note that a morphism on $A^*$ is
entirely determined by the images of the letters of $A$.

Given a set $X$ of words over $A$, and given a morphism $f$ on $A^*$,
we denote by $f(X)$ the set $\{f(w) \mid w \in X\}$.

Given an integer $L$, $f$ is \textit{$L$-uniform}
if $|f(a)| = L$ for every letter $a$ in $A$.
A morphism $f$ is \textit{uniform} if it is $L$-uniform for some integer $L \geq 0$.

A morphism $f$ on $A^*$ is overlap-free if the image of any overlap-free word by $f$ is also overlap-free. 
In other words, an overlap-free morphism is a morphism that \textit{preserves} overlap-free words.
For instance, the \textit{empty morphism} $\epsilon$
($\forall a \in A$, $\epsilon(a) = \varepsilon$) or
the \textit{identity endomorphism} $Id$
($\forall a \in A$, $Id(a) = a$) 
are overlap-free.

A morphism $f$ on $A^*$ is \textit{$k$-power-free}  ($k \geq 2$)
if $f(w)$ is $k$-power-free for every $k$-power-free word $w$ in $A^+$.

We say that a morphism is \textit{non-erasing} if, for all letters $a \in A$,
$f(a) \neq \varepsilon$.
The {empty morphism} $\epsilon$
is the only morphism that is both erasing and overlap-free or
$k$-power-free ($k \geq 2$).
Indeed, for any erasing morphism $f \neq \epsilon$ defined on $A^*$,
there exist two different letters $a$ and $b$ in $A$ (remember $\card{A} \geq 2$)
such that $f(a) \neq \varepsilon$, $f(b) = \varepsilon$, and so
$f(aba^{k-1})$ contains a $k$-power and especially an overlap when $k=3$ with 
$aba^{k-1}$ $k$-power-free and $abaa$ overlap-free.
From now, since the case $f=\epsilon$ is of little interest, we always assume that $f \neq \epsilon$.

A morphism on $A^*$ is called \textit{prefix} (resp. \textit{suffix})
if, for all different letters $a$ and $b$ in $A$, 
the word $f(a)$ is not a prefix (resp. not a suffix) of $f(b)$.
A prefix (resp. suffix) morphism is non-erasing.
A morphism is \textit{bifix} if it is prefix and suffix.

A non-erasing morphism on $A^*$ is called \textit{strongly prefix} (resp. \textit{strongly suffix})
if, for all different letters $a$ and $b$ in $A$, 
the words $f(a)$ and $f(b)$ do not start (resp. do not end) with the same letter.
A morphism is \textit{strongly bifix} if it is strongly prefix and strongly suffix.

\begin{rqe}\label{Rem2bif}
If $f$ is a strongly bifix morphism on $A^*$, for any letter $x \in A$,
$x$ is entirely determined by the first letter or the last letter of $f(x)$.
\end{rqe}

\begin{rqe}\label{Rem1bif}
By definition, a strongly bifix morphism is a bifix morphism.
\end{rqe}

\begin{lemma}\label{LemAbifix}
An overlap-free morphism is  a strongly bifix morphism.
\end{lemma}

\begin{proof}

Let $f$ be a morphism from $A^*$ to $B^*$ and, by contraposition, suppose that $f$ is not a strongly bifix.

If $f$ is not strongly prefix, let $a$ and $b$ be two different letters in $A$ such that $f(a)$ and $f(b)$ starts
with the same letter $x$. Let $\alpha$ and $\beta$ be the words such that $f(a)=x\alpha$ and
$f(b)=x\beta$.
The image of the overlap-free word $aab$ by $f$ contains the overlap
$x\alpha x\alpha x$:  $f$ is not overlap-free.

If $f$ is not strongly suffix, on the same way, we get that the image of a word of the form $baa$
contains an overlap.
\end{proof}

Given a morphism $f$ on $A$, the \textit{mirror morphism} $\tilde{f}$
of $f$ is defined for all words $w$ in $A^+$, by
$\tilde{f}(w)=\widetilde{f(\tilde{w})}$. 
In particular, $\tilde{f}(a)=\widetilde{f(a)}$ for every letter $a$ in $A$.

Let us recall that a word $w$ is overlap-free (resp. $k$-power-free)
if and only if $\tilde{w}$ is overlap-free (resp. $k$-power-free).
As a direct consequence, a morphism $f$ is overlap-free (resp. $k$-power-free) if and only if 
$\tilde{f}$ is overlap-free (resp. $k$-power-free).

Proofs of the three following lemmas are left to the reader.

\begin{lemma}
\label{LemBif1}
Let $f$ be a bifix morphism on $A^*$ and 
let $u$, $v$, $w$, and $t$ be words in $A^*$.\\
The equality $f(u)=f(v)p$ 
where $p$ is a prefix of $f(w)$
implies $u=vw'$ for a prefix $w'$ of $w$ such that $f(w')=p$.\\
Symetrically, the equality $f(u)=sf(v)$ 
where $s$ is suffix of $f(t)$
implies $u=t'v$ for a suffix $t'$ of $t$ such that $f(t')=s$.

\end{lemma}

%

%

\begin{lemma}
\label{1pref}
Let $f$ be a strongly prefix morphism on $A^*$,
let $u$ and $v$ be words in $A^*$,
and let $a$ and $b$ be letters in $A$.
Furthermore,
let $p_1$ (resp. $p_2$) be a prefix of $f(a)$ (resp. of $f(b)$).
If $p_1p_2 \neq \varepsilon$, $(p_1;p_2) \neq (\varepsilon;f(b))$ and if $(p_1;p_2) \neq (f(a);\varepsilon)$ 
then
the equality $f(u)p_1=f(v)p_2$ implies $u=v$ and $a=b$.

\end{lemma}

\begin{lemma}
\label{1suff}
Let $f$ be a strongly suffix morphism on $A^*$,
let $u$ and $v$ be words in $A^*$, 
and let $a$ and $b$ be letters in $A$.
Furthermore,
let $s_1$ (resp. $s_2$) be a suffix of $f(a)$ (resp. of $f(b)$).
If $s_1s_2 \neq \varepsilon$, $(s_1;s_2) \neq (\varepsilon;f(b))$ and if $(s_1;s_2) \neq (f(a);\varepsilon)$ 
then
the equality $s_1f(u)=s_2f(v)$ implies $u=v$ and $a=b$.

\end{lemma}

\begin{definition}

A morphism $f$ from $A^*$ to $B^*$ is a \textit{ps-morphism} (Ker\"anen \cite{Ker1986} called 
$f$ a ps-code)
if and only if the equalities \\ \centerline{$f(a) = ps$, $f(b) = ps'$ and
$f(c) = p's$} with $a,b,c \in A$ (possibly $c = b$) and $p$, $s$, $p'$,
$s' \in B^*$ imply $b = a$ or $c = a$.

\end{definition}

%
%
%
%

\begin{rqe}\label{Rem1bifPS}

By definition, a strongly bifix morphism is a ps-morphism.
\end{rqe}

\begin{lemma}{\rm \cite{Ker1986,Lec1985}}
\label{lemmeSPKPS}
If $f$ is not a ps-morphism then $f$ is not a $k$-power-free morphism for all integers $k \geq 2$.
\end{lemma}

\begin{lemma}\label{LemOverPS}
An overlap-free morphism ($\neq \epsilon)$ is  a ps-morphism.
\end{lemma}

\begin{proof}

Let $f$ be a morphism from $A^*$ to $B^*$ and, by contraposition, suppose that $f$ is not a ps-morphism.

Let $a$, $b$, and $c$ be letters of $A$ with $b \neq a$ and $c \neq a$
such that $f(a) = ps$, $f(b) = ps'$, and $f(c) = p's$ for some words $p$, $p'$, $s$, and $s'$.

If $p \neq \varepsilon$, then $f(aab)$ contains $pspsp$. 
If $s \neq \varepsilon$, then $f(caa)$ contains $spsps$. 
In these both cases, we get that $f$ is not an overlap-free morphism.

If $f(a) = \varepsilon$, since $f \neq \epsilon$, let $d$ be a letter of $A$ such that $f(d) \neq \varepsilon$. 
Therefore, we get that $f(dadaad) = f(d)^3$, meaning that $f$ is not overlap-free.
\end{proof}

Let us assume that $f(\overline{w})=pu^ks$ for a factor $\overline{w}$ of a word $w$,
a non-empty word $u$ and an integer $k \geq 2$,
Let us also assume that $\overline{w}$ contains a factor $w_0$ such that $|f(w_0)|=|u|$.
If $f$ is a ps-morphism, Lemma~\ref{synckp} states
that $\overline{w}$ necessarily contains a $k$-power $w'^k$ such 
that $f(w')$ is a conjugated word of $u$.
We will say that $f(w)$ contains a synchronised $k$-power $u^k$
or that $f(w)$ and $u^k$ are synchronised. More precisely:

\begin{defin}\label{defsync}
Let $k \geq 2$ be an integer.
Let $f$ be a morphism from $A^*$ to $B^*$, $w$ be a word in $A^+$, and $u$ be a word in $B^+$
such that $f(w)$ contains the $k$-power $u^k$.
Let $\overline{w}$ be a shortest factor of $w$ whose image by $f$ contains $u^k$,
i.e., $f(\overline{w})=pu^ks$ with $|p|<|f(\sub{\overline{w}}{1})|$ and 
$|s|<|f(\sub{\overline{w}}{|\overline{w}|})|$.

We say that $f(w)$ and $u^k$ are synchronised if there exist
three words $w_0$, $w_1$, and $w_2$ such that $|f(w_0)|=|u|$ and $\overline{w}=w_1w_0w_2$
with $p=\varepsilon$  if $w_1=\varepsilon$, and
$s=\varepsilon$ if $w_2=\varepsilon$.

\end{defin}

\begin{lemma}\label{synckp}{\rm \cite{Wla2015a}}

Let $k \geq 2$ be an integer.
If $f$ is a ps-morphism and if $f(w)$ contains a synchronised $k$-power  then 
$w$ contains a $k$-power.

\end{lemma}

\begin{rqe}\label{rsynckp}{\rm \cite{Wla2015a}}

More precisely, the word $\overline{w}$ starts or ends with a $k$-power 
whose image by $f$ is a conjugated of the synchronised $k$-power.

\end{rqe}

\begin{defin}\label{defsyncOverlap}
Let $f$ be a morphism from $A^*$ to $B^*$, $w$ be a word in $A^+$, $a$ be a letter in $A$ and 
$u$ be a word in $B^*$ such that $f(w)$ contains the overlap $auaua$.
Let $\overline{w}$ be a shortest factor of $w$ whose image by $f$ contains $auaua$,
i.e., $f(\overline{w})=p auaua s$ with $|p|<|f(\sub{\overline{w}}{1})|$ and 
$|s|<|f(\sub{\overline{w}}{|\overline{w}|})|$.

We say that $f(w)$ and $auaua$ are synchronised if there exist
three words $w_0$, $w_1$, and $w_2$ such that $|f(w_0)|=|au|$ and $\overline{w}=w_1w_0w_2$
with $p=\varepsilon$  if $w_1=\varepsilon$, and
$s=\varepsilon$ if $w_2=\varepsilon$.
\end{defin}

\begin{rqe}
Let us note that, $f(w)$ and $auaua$ are synchronised only if 
$f(w)$ and $(au)^2$ are synchronised. 

\end{rqe}

\begin{lemma}\label{LemSyncOverlap}

If $f$ is a strongly bifix morphism and if $f(w)$ contains a synchronised overlap then 
$w$ contains an overlap.
More precisely, the word $\overline{w}$ (see Definition~\ref{defsyncOverlap}) starts or ends with an overlap 
whose image by $f$ is a conjugated word of the synchronised overlap.

\end{lemma}

\begin{proof}

Let $a$ be the letter and let $u$ be the word such that $f(w)$ and $auaua$ are synchronised,
let $\overline{w}$ be the shortest factor of $w$ whose image by $f$ contains $auaua$, and 
let $w_0$ be a factor of $\overline{w}$ such that $|f(w_0)|=|au|$.

There exist 
a prefix $p$ of
$f(\sub{\overline{w}}{1})$ and a suffix $s$
of $f(\sub{\overline{w}}{|\overline{w}|})$ 
such that $f(\overline{w})=pauauas$ with 
$|p| < |f(\sub{\overline{w}}{1}|$ and $|s| < |f(\sub{\overline{w}}{|\overline{w}|}|$.
Let $w_1$ and $w_2$ be the words  such that $w=w_1w_0w_2$.


If $w_1=\varepsilon$, i.e., $\overline{w}$ starts with $w_0$, then $p=\varepsilon$ and 
$f(\overline{w})$ starts with $au$. 
It implies that $au=f(w_0)$. By Lemma~\ref{LemBif1}, it follows that $\overline{w}$ starts with $w_0^2$.
Let $x$ be the first letter of $w_0$. Let $w_0'$ and $w'$ be the words such that $w_0=xw_0'$
and $w=xw_0'xw_0'w'$.
Since $f(w')$ starts with $a$ and since $f$ is strongly bifix, we get that $w'$ starts with $x$
(see Remark~\ref{Rem2bif})
and that $w$ starts with the overlap $xw_0'xw_0'x$.

If $w_2=\varepsilon$, i.e., $\overline{w}$ ends with $w_0$, then, in a similar way,
we obtain that $\overline{w}$ ends with $yw_0''yw_0''y$ where $y$ is the last letter 
of $w_0$ and $w_0''$ is the word such that $w_0=w_0''y$.

From now, let us assume that $w_1$ and $w_2$ are non-empty, i.e.,
$w_0$ is an internal factor of $\overline{w}$.
Since $|pa| \leq |f(\sub{w_1}{1}|$, $|as| \leq |f(\sub{w_2}{|w_2|}|$,
it implies that $f(w_0)$ is an internal factor of $uau$.
In particular, since $f(w_0)=|au|$,  it means that $f(w_0)$ and $au$ are conjugated words.
More precisely, there exist four words $u_1$, $u_2$, $u_3$ and $u_4$
such that
$f(w_1)=pau_1$, $f(w_2)=u_4as$, $f(w_0)=u_2au_3$ with $u=u_1u_2=u_3u_4$.

Since $|f(w_0)|=|u_2au_3|=|au|=|au_1u_2|$, we get that $u_1=u_3$ and $u_2=u_4$, that is,
$f(w_1)=pau_1$, $f(w_2)=u_2as$ and $f(w_0)=u_2au_1$.

Let $t$ be the greatest common suffix of $w_0$ and $w_1$.
Since $f$ is strongly bifix, we have $|f(t)| \geq |au_1|$.
Since $|p| \leq |f(\sub{w_1}{1}|$, by Lemma~\ref{1suff}, we get that $w_1=t$ is a suffix of $w_0$.
In particular, $u_2$ ends with $p$. There exist a word $w_0'$ and a word $u_2'$
such that $w_0=w_0'w_1$, $u_2=u_2'p$ and $f(w_0')=u_2'$.

Since $f(w_2)=u_2as=f(w_0')pas$ avec $|pa| \leq |f(\sub{w_1}{1}|$ and $|as| \leq |f(\sub{w_2}{|w_2|}|$
and since $f$ is strongly bifix, we get that 
there exist a letter $x$ such that $f(x)=pas$ (see Remark~\ref{Rem2bif}).
By Lemma~\ref{1pref}, it implies that $w_2=w_0'x$.
It also means that $w_1$ starts with $x$. So there exist a word $w_1'$ such that $w_1=xw_1'$.
It follows that $\overline{w}=xw_1'w_0'xw_1'w_0'x$ is an overlap
with $|f(xw_1'w_0')|=|f(w_1w_0')|=|au|$.
\end{proof}

\begin{lemma}\label{imagepure}{\rm \cite{Wla2015a}}

Let $k \geq 4$ be an integer.\\
The image of a pure $k$-power by a $k$-power-free morphism is also a pure $k$-power.

\end{lemma}

\begin{rqe}\label{RemImagePasPure}

Let $f$ be an overlap-free morphism from $A^*$ to $B^*$.
If $auaua$ is a pure overlap and if $|f(a)| \geq 2$ then $f(auaua)$ is not pure.
Indeed, if $f(a)=a_1A_1$ with  $a_1 \in A$ and $A_1 \in A^+$, we get that 
$a_1A_1f(u)a_1A_1f(u)a_1$ is a proper factor (prefix) of $f(auaua)$.

\end{rqe}

%
%
%

%
%



\section{\label{outils2}
        Reduction of a power}

\subsection{\label{spow} About overlap-free morphisms}

\begin{lemma}\label{LemConjImOver}

Let $f$ be a strongly bifix morphism from $A^*$ to $B^*$.

We assume that there exists a letter $a$ in $A$ and two words $T \in A^+$ and $u \in A^*$
such that $auaua$ is a pure overlap and 
$f(T)=\pi_1f(auaua)\sigma_2$ with $\pi_1, \sigma_2 \in B^*$ satisfying
$|\pi_1|<|f(\sub{T}{1})|$ and $|\sigma_2|<|f(\sub{T}{|T|})|$.

Then, either $f$ is not overlap-free, or 
$T$ starts with a pure overlap $bvbvb$ and ends with a pure overlap $ctctc$
such that $|f(bv)|=|f(ct)|=|f(au)|$.
\end{lemma}

\begin{rqe}

If $au$ and $bv$ are conjugated words, then the same holds for $f(au)$ and $f(bv)$ (and, trivially,
$|f(bv)|=|f(au)|$). The converse does not hold.

\end{rqe}

\begin{rqe}\label{RemConjImOver}

If $T$ starts with $bvbvb$ then $|\pi_1|<|f(b)|$.
If $T$ ends with $bvbvb$ then $|\sigma_2|<|f(b)|$.
So either $f(bv)=f(au)$ or  $f(bv)$ is a conjugated word of $f(au)$.
So, one of the word $f(bv)$ or $f(vb)$ is an internal factor of $f(au)f(au)$ with 
$|f(bv)|=|f(vb)|=|f(au)|$. 

\end{rqe}

\begin{proof2}{Lemma~\ref{LemConjImOver}}

Since $f$ is a ps-morphism, we have $f$ non-erasing and injective.

If $T$ is an overlap-free word, then $f$ is not overlap-free, and the proof is complete in this case.
So, we assume that $T = T_1bvbvbT_2$ where $b \in A$, $v, T_1, T_2 \in A^*$, with the additional condition 
that $bvbvb$ is a pure overlap.

Therefore, we have $f(T_1)f(b)f(v)f(b)f(v)f(b)f(T_2) = \pi_1f(a)f(u)f(a)f(u)f(a)\sigma_2$.

Let $\alpha_1$ be the first letter of $f(a)$ and $\alpha_2$ be the last letter of $f(a)$.
Let $A_1$ and $A_2$ be the words in $B^*$ such that $f(a)=\alpha_1A_1=A_2\alpha_2$.

Let $\beta_1$ be the first letter of $f(b)$ and $\beta_2$ be the last letter of $f(b)$.
Let $B_1$ and $B_2$ be the words in $B^*$ such that $f(b)=\beta_1B_1=B_2\beta_2$.

\setcounter{comptnivun}{1}
{\textbf{\textit{Step \thecomptnivun}}}:
If $T$ starts with an overlap $bvbvb$ such that $|f(bv)|=|f(au)|$
then it ends with an overlap $ctctc$ such that $|f(ct)|=|f(au)|$
and vice versa.

The case where $T$ ends with an overlap is the mirror case of the 
case where $T$ starts with an overlap. 
We only deal with this last case.

Let us assume that $T=bvbvbT'$ with $|f(bv)|=|f(au)|$.

If $\pi_1=\varepsilon$, then, by Lemma~\ref{1pref}, we get $au=bv$. 
Since $|\sigma_2|<|f(T')|$, it implies that $\sigma_2=\varepsilon$ and $T'=\varepsilon$,
i.e., $T=auaua$ (with obviously $|f(au)|=|f(ua)|=|f(bv)|$).

If $\pi_1 \neq \varepsilon$, let $bv_1$ be the shortest prefix of $bv$ such that $|f(bv_1)|>|\pi_1f(a)|$.
Let $\sigma'$ be the word such that $f(bv_1)=\pi_1f(a)\sigma'$ with $|\sigma'|<|f(\sub{v_1}{|v_1|})|$ and
let $v_2$ be the word such that $v=v_1v_2$.
Since $|\pi_1f(au)|=|f(bv)|+|\pi_1|=|\pi_1f(a)\sigma'| +|f(v_2)|+|\pi_1|$, it follows that 
$f(u)=\sigma'f(v_2)\pi_1$.
It implies that $f(bT')=\pi_1f(a) \sigma_2$.
Since $f$ is strongly bifix, $f(bv_1)=\pi_1f(a)\sigma'$,
$|\sigma'|<|f(\sub{v_1}{|v_1|})|$ and $|\sigma_2|<|f(\sub{T'}{|T'|})|$, by Lemma~\ref{1pref},
we get that $T'=v_1$. Let $c$ be the last  letter of $bv_1$ ($\neq \varepsilon$) and 
let $v_3$ be the word such that $bv_1=v_3c$.
We get that $T=v_3cv_2v_3cv_2v_3c$ with $|f(v_2v_3c)|=|f(bv)|=|f(au)|$.

Let us note that, in fact, $v_2v_3c$ is a conjugated word of $bv$.

\addtocounter{comptnivun}{1}
{\textbf{\textit{Step \thecomptnivun}}}:
Elementary cases

\quad \textit{Case 1.a}: $|f(T_1b)| = |\pi_1f(a)|$

In this case, $f(b)$ ends with $\alpha_2$. If $b \neq a$ then
$f(baa)$ ends with $\alpha_2 A_2\alpha_2 A_2\alpha_2$. This means that $f$ is not overlap-free.
If $b=a$ then, since $|\pi_1|<|f(\sub{T}{1})|$, we obtain $T_1= \varepsilon=\pi_1$.
From Lemma~\ref{LemBif1}, the equality $(f(T)=)f(bvbvbT_2)=f(auaua)\sigma_2$ with
$|\sigma_2|<|f(\sub{T}{|T|}|$ implies that $\sigma_2=\varepsilon$ and that $bvbvbT_2=auaua$.
Since $auaua$ is a pure overlap,
we therefore have $T_2=\varepsilon$ and $au=bv$.



\quad \textit{Case 1.b:} $|f(bT_2)| = |f(a)\sigma_2|$

As in case 1.a, we obtain that $T$ ends with
$ua=vb$ when $a=b$. And $f(aab)$ starts with $\alpha_1A_1\alpha_1A_1\alpha_1$ when $a\neq b$,
i.e., $f$ is not overlap-free.

\quad \textit{Case 2.a}: $|\pi_1f(a)| \leq |f(T_1)|$

Since $|\sigma_2|<|f(\sub{T}{|T|})|$,
the word $f(uaua)$ contains the overlap $\beta_1B_1f(v)\beta_1B_1f(v)\beta_1$
with $uaua$ overlap-free: $f$ is not overlap-free.

\quad \textit{Case 2.b}: $|f(a)\sigma_2| \leq |f(T_2)|$ 

As in case 2.a and since $|\pi_1|<|f(\sub{T}{|1|})|$, we get that
the word $f(auau)$ contains the overlap $\beta_2f(v)B_2\beta_2f(v)B_2\beta_2$
with $auau$ overlap-free:~$f$ is not overlap-free.

\quad \textit{Case 3.a}: $|\pi_1f(a)| > |f(T_1b)|$ and $T_2 =\varepsilon$

It means that $|\pi_1A_2| \geq |f(T_1b)|$. So
the overlap $\alpha_2f(u)A_2\alpha_2f(u)A_2\alpha_2$ is a factor of $f(vbvb)$ with
$vbvb$ overlap-free:~$f$ is not overlap-free.

\quad \textit{Case 3.b}: $|f(a)\sigma_2|>|f(bT_2)|$ and $T_1=\varepsilon$

As Case 3.a, we obtain that the overlap $\alpha_1A_1f(u)\alpha_1A_1f(u)\alpha_1$ is a factor of $f(bvbv)$ 
with $bvbv$ overlap-free:~$f$ is not overlap-free.

\addtocounter{comptnivun}{1}
{\textbf{\textit{Step \thecomptnivun}}}:  
Other cases

If $f(bvbvb)$ and $f(auaua)$ have a common factor of length greater than or equal to $|f(au)| + |f(bv)|$ then,
according to the corollary~\ref{Kera}, there exist two integers $i,j \geq 1$ and two words $t_1,t_2$ such that $f(au) = (t_1t_2)^i$ and $f(bv) = (t_2t_1)^j$, where $t_1t_2$ and $t_2t_1$ are primitive words.

If $i \geq 2$, then $f(auau)$ contains a cube, and therefore an overlap with $auau$ overlap-free. 
Similarly, if $j \geq 2$, we obtain that $f$ is not overlap-free.
Having $i=j=1$ means that $f(au)$ and $f(bv)$ are conjugated words.
Since $|\pi_1|<|f(\sub{T}{1}|$ and $|\sigma_2|<|f(\sub{T}{|T|})|$,
if we have $T_1 = \varepsilon$, or $T_2 = \varepsilon$ then it ends the proof.

So, from now, when  $T_1 = \varepsilon$, or $T_2 = \varepsilon$, we will assume that any common 
factor of $f(bvbvb)$ and $f(auaua)$ is of length less than $|f(au)| + |f(bv)|$.

\quad \textit{Case 4.a}: $|f(T_1)| < |\pi_1f(a)| <|f(T_1b)|$ and $T_1 \neq \varepsilon$

If $T_2= \varepsilon$, then $f(uaua)$ and $f(bvbv)$ are common factor of $f(bvbvb)$ and $f(auaua)$
(of length greater than or equal to $|f(au)| + |f(bv)|$).
So, as previously said, it ends the proof.

If $T_2 \neq \varepsilon$, since $|\sigma_2|<|f(T_2)|$,
the word $f(uaua)$ contains the overlap $\beta_2f(v)B_2\beta_2f(v)B_2\beta_2$
with $uaua$ overlap-free:~$f$ is not overlap-free.

\quad \textit{Case 4.b}: $|f(bT_2)| < |f(a)\sigma_2|<|f(T_2)|$ and $T_2 \neq \varepsilon$ 

This case is solved as Case 4.a.

\quad  \textit{Case 5:} $|\pi_1f(a)| < |f(b)|$ and $|f(a)\sigma_2| < |f(b)|$ with $T_1=T_2=\varepsilon$ 

As we previously said, if $f(bvbvb)$ and $f(auaua)$ have a common factor of length greater or equal 
to $|f(au)| + |f(bv)|$ then it ends the proof.

There exist two non-empty words $x$ and $t$ such that
$f(b)=\pi_1f(a)t=xf(a)\sigma_2$. Since $f(auaua)=f(a)tf(vbv)xf(a)$ is a common factor
of $f(bvbvb)$ and $f(auaua)$, we have
$|f(auaua)|+|f(a)tf(vbv)xf(a)| < 2|f(au)| + 2|f(bv)|$, i.e.,
$|f(b)|>3|f(a)|+|x|+|t|$. Therefore, there exists a word $Y$ such that $|Y|>|f(a)|$, $\pi_1=xf(a)Y$,
and $\sigma_2=Yf(a)t$, that is, $f(b)=xf(a)Yf(a)t$. 
Since $f(uau)=tf(vbv)x$, we obtain that $|f(u)|=|tf(v)x| +\frac{1}{2}|aY|> |t|+|x|+|f(v)f(a)|$.
So, there exist two words $y$ and $z$ such that $Y=yf(a)z$, i.e., $f(b)=xf(a)yf(a)zf(a)t$ and 
$f(u)=(tf(v)x) \, (f(a)y) = (zf(a))\, (tf(v)x)$.
From Case~\ref{Lotcase1} of Proposition~\ref{Lothaire}, we get  that
there exist two words $r$ and $s$ in $A^*$, and an integer $n$ such that $tf(v)x=r(sr)^{n}$, $zf(a)=rs$,
and $f(a)y=sr$.
It follows that $f(bvb)$ contains $zf(a)tf(v)xf(a)y=r(sr)^{n+2}$ for an integer $n$
(since $tx \neq \varepsilon$, if $n=0$ then $r \neq \varepsilon$)
and so an overlap. Since $bvb$ is overlap-free, $f$ is not overlap-free.

\quad \textit{Case 6}: $|\pi_1f(a)| > |f(T_1b)|$ and $|f(a)\sigma_2|>|f(bT_2)|$ with $T_1 \neq \varepsilon$
and $T_2 \neq \varepsilon$

There exist two non-empty words $\sigma_1$ and $\pi_2$ such that
$f(T_1)=\pi_1 \sigma_1$ and $f(T_2)=\pi_2 \sigma_2$.
Since $|\pi_1|<|f(\sub{T}{1}|$, $|\sigma_2|<|f(\sub{T_2}{|T_2|})|$, $f(b) \neq \varepsilon$ and
$f$ strongly bifix morphism, the word $\sigma_1$ is a proper prefix
of $f(a)$ and the word $\pi_2$ is a proper suffix of $f(a)$. 

If $f(bvbvb)$ and $f(auaua)$ have a common factor of length greater or equal to $|f(au)| + |f(bv)|$ then,
$|f(au)|=|f(bv)|$.
So $|f(T_1bv)|=|\pi_1f(au)\sigma_1|$ and  it follows that $\sigma_1$ is a suffix of $f(v)$. 
By Lemma~\ref{LemBif1}, since $|\pi_1|<|f(\sub{T_1}{1})|$, we get that
$T_1$ is a suffix of $v$. Let $v'$ be the word such that $v=v'T_1$.
The overlap $T_1bv'T_1bv'T_1$ is a prefix of $T$.
Furthermore, $T_1bv'$ and $bv$ are conjugated words.
It follows that $|f(T_1bv')|=|f(au)|$.

In the case of any common factor of $f(bvbvb)$ and $f(auaua)$ is of length less than $|f(au)| + |f(bv)|$,
we obtain the equation $\sigma_1 f(bvbvb) \pi_2=f(auaua)$ with $|\sigma_1 f(b)|<|f(a)|$
and $|f(b) \pi_2|<|f(a)|$. This case is solved as Case 5, by exchanging the roles of $au$ and $bv$.
\end{proof2}


\subsection{\label{cond}
        Equations of reduction}

The reduction technique described in Lemma~\ref{Reduc1} and using Lemma~\ref{PcorEquas} 
is the central idea in the proof of Proposition~\ref{mainresult}.

\begin{lemma} \label{PcorEquas}{\rm \cite{Wla2015a}}

Let $\alpha_1,\alpha_2$, $\beta_1,\beta_1',\beta_2$, $\gamma_1, \gamma_2$
be words over an alphabet $B$ such that
$|\beta_1|=|\beta_2| \neq 0$, $\beta_1'$ is a proper suffix of $\beta_1$, and
$0 \leq |\alpha_2|- |\alpha_1| \leq |\beta_1'|$.\\
Under these hypotheses, the equality $\alpha_2 \beta_2 \gamma_2 = \alpha_1\beta_1'\beta_1\gamma_1$ 
implies  $\alpha_2 \gamma_2=\alpha_1\beta_1'\gamma_1$.

\end{lemma}


\begin{lemma} \label{Reduc1}{\rm \cite{Wla2015a}}

Let ${\kappa} \geq 3$ be an integer.
Let $f$ be a morphism from $A^*$ to $B^*$.
Let $(w_i)_{i=1..{\kappa}+1}$ and $(x_i)_{i=1..{\kappa}}$ be words in $A^*$
such that $|f(x_i)|=|f(x_j)|\neq 0$ for all integers $i,j$ in $[1,{\kappa}]$.

We denote by $w$ the word $w_1x_1...w_{\kappa}x_{\kappa}w_{{\kappa}+1} $.

We assume that there exist words 
$u$, $p$, $s$, $(X_i)_{i=1..{\kappa}}$, and $(Y_i)_{i=1..{\kappa}}$ in $B^*$ such that
$f(w_1)= pX_1$, $f(w_{{\kappa}+1})=Y_{{\kappa}}s$, and
$f(w_i)=Y_{i-1}X_{i}$ for all $2 \leq i \leq {\kappa}$.
Moreover, we assume that, for all integers $i$ in $[1,{\kappa}]$, we have 
$u=X_if(x_i)Y_i$. It means that $f(w)=pu^{\kappa}s$.

Let us also assume that there exists an integer $q$ such that, for every integer $i$ in $[1,{\kappa}]$, 
$0 \leq |X_q|- |X_i| \leq |X''_q|$ where $X''_q$ is a common suffix of $X_q$ and $f(x_q)$.

Then the word $\check{w}=w_1w_2...w_{\kappa}w_{{\kappa}+1}$ satisfies
$f(\check{w})=p\check{u}^{\kappa}s$ with
$\check{u}=X_iY_i$ for every integer $i$ in $[1,{\kappa}]$.

In particular, $f(\check{w})$ and $\check{u}^{\kappa}$ are synchronised only if $f(w)$ and $u^{\kappa}$ are synchronised.

\end{lemma}

The situation described in Figure~\ref{reduc1a} is an example of a case
where the hypotheses of Lemma~\ref{Reduc1} hold. 

\begin{figure}[!ht]
\centering
\includegraphics[width=12.5cm]{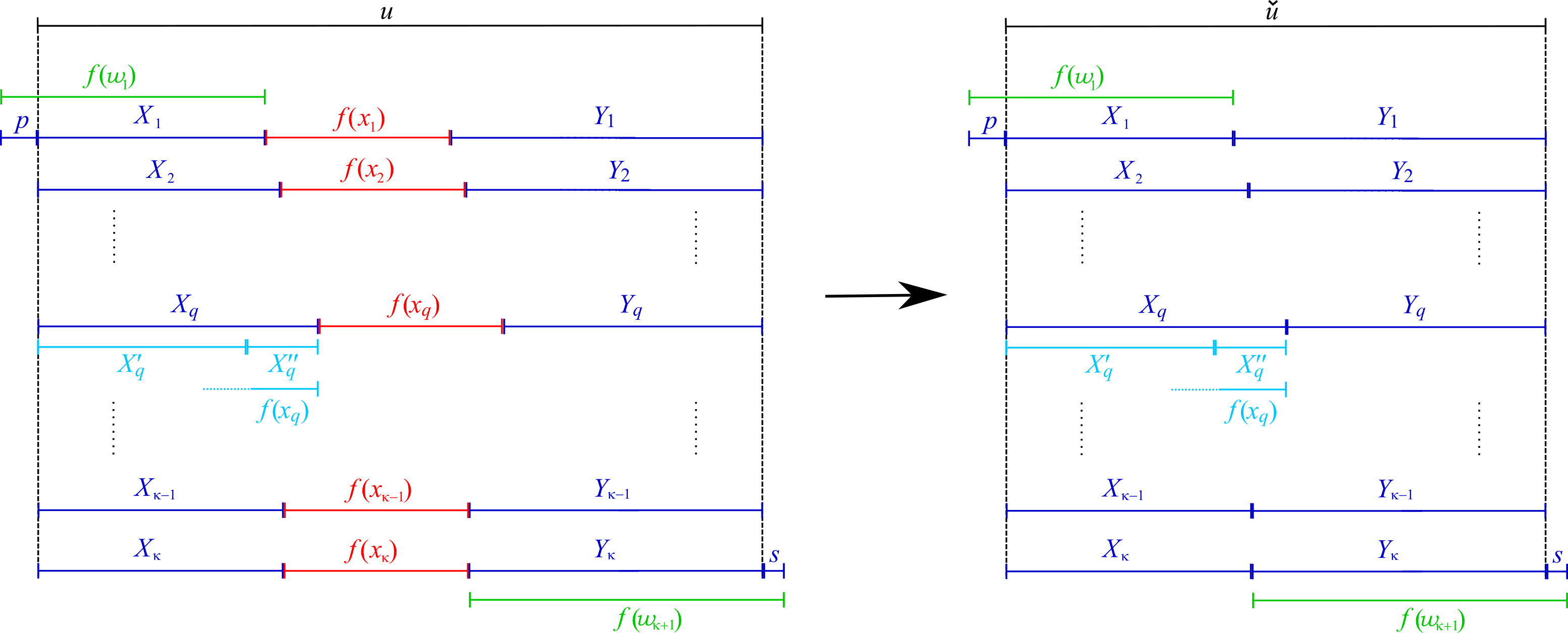}
\caption{Reduction of a power}
\label{reduc1a}
\end{figure}

We say that we have reduced $w$.

Let us note that $p$ is not necessarily a prefix of $f(\sub{w_1}{1})$ and 
$s$ is not necessarily a suffix of $f(\sub{w_{\kappa}}{|w_{\kappa}|})$.

\begin{rqe}\label{RemReduc1}

Let us remark that some hypotheses of Lemma~\ref{Reduc1} are satisfied when $f(x_i)=f(x_q)$ or
when $f(x_i)$ is an internal factor of 
$X_q''f(x_q)$ so of $f(x_qx_q)$ with $|f(x_i)|=|f(x_q)|$, i.e., $f(x_i)$ is a conjugated word $f(x_q)$.

\end{rqe}

Figure~\ref{remcas4} deals with Point~\ref{RemReducC1} of Remark~\ref{RemReduc} and
Figure~\ref{remcas5} deals with Point~\ref{RemReducC2} of Remark~\ref{RemReduc}.
In addition to the situation described in Figure~\ref{reduc1a}, we will mostly use these two other
points in the proof of Proposition~\ref{mainresult}.
%

\begin{figure}[!ht]
\centering
\includegraphics[width=12.5cm]{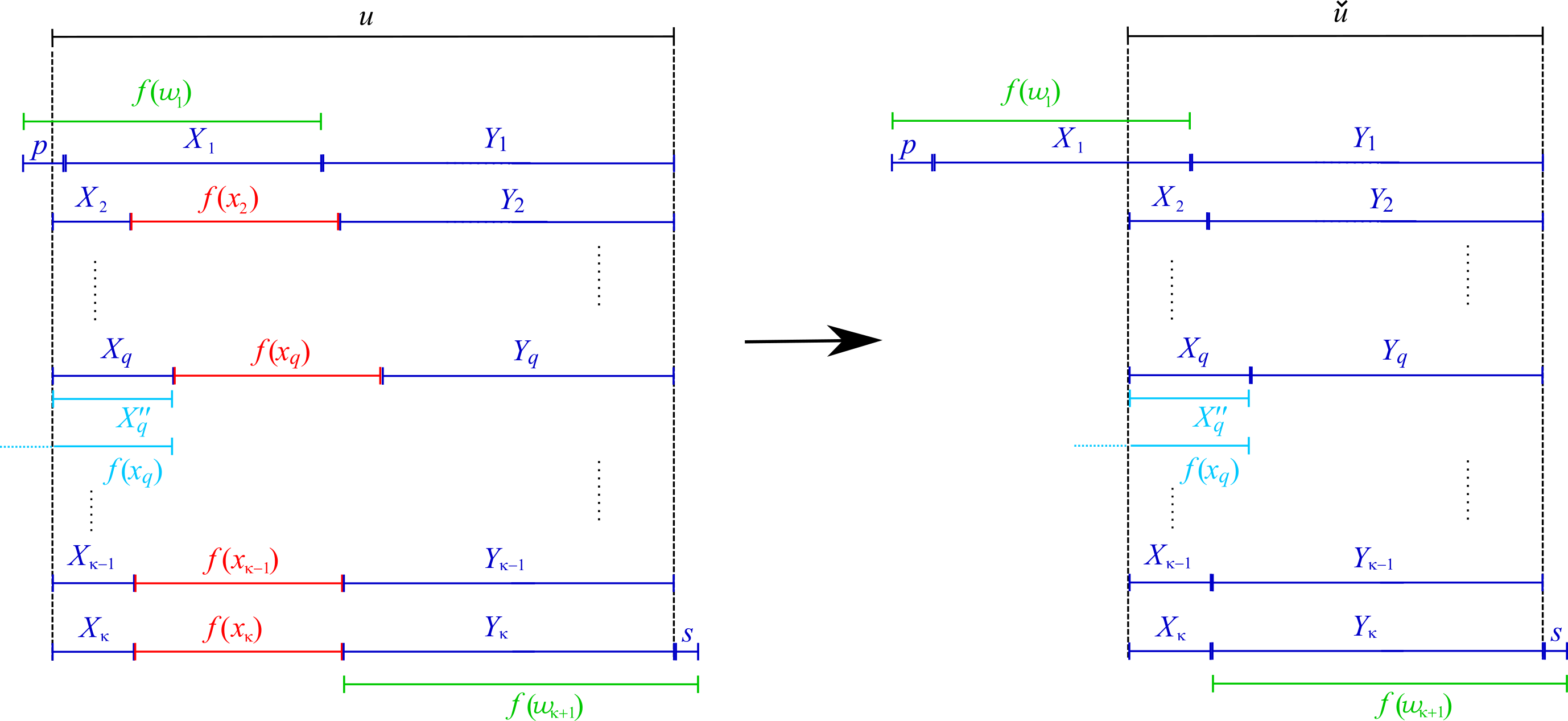}
\caption{Point~\ref{RemReducC1} of Remark~\ref{RemReduc}}
\label{remcas4}
\end{figure}

\begin{figure}[!ht]
\centering
\includegraphics[width=12.5cm]{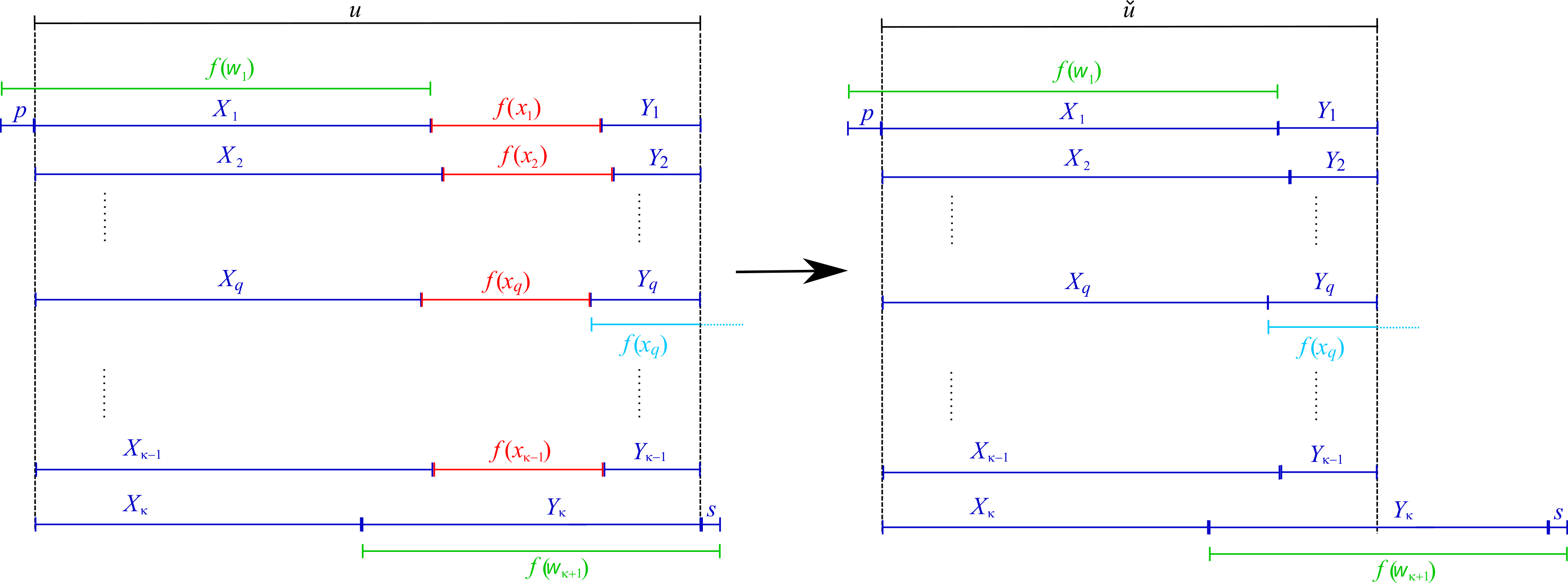}
\caption{Point~\ref{RemReducC2} of Remark~\ref{RemReduc}}
\label{remcas5}
\end{figure}

\begin{rqe}\label{RemReduc}

\begin{enumerate}

\item Using the mirror image and exchanging $|X_q|$ the maximum of $|X_i|$ by the maximum $|Y_q|$ of $|Y_i|$ 
(i.e., $|X_q|$ is the minimum of $|X_i|$), the condition 
"$0 \leq |X_q|- |X_i| \leq |X''_q|$ where $X''_q$ is a common 
suffix of $X_q$ and $f(x_q)$" of Lemma~\ref{Reduc1} can be replaced by "$0 \leq |Y_q|- |Y_i| \leq |Y'_q|$ 
where $Y'_q$ is a common prefix of $Y_q$ and $f(x_q)$".

\item A prefix $u_1$ of $u$ is also a prefix of $\check{u}$ if $|u_1|<|X_q|$, and 
a suffix $u_2$ of $u$ is also a suffix of $\check{u}$ if $|u_2|< \max |Y_j|$.

\item \label{RemPasComp} 

If, instead of $u=X_{\kappa}f(x_{\kappa})Y_{\kappa}$, we only have
that  $X_{\kappa}f(x_{\kappa})Y_{\kappa}$ is a prefix of $u$ 
then $f(\check{w})=p\check{u}^{\kappa-1}X_{\kappa}Y_{\kappa}s$  with 
$X_{\kappa}Y_{\kappa}$ prefix of $\check{u}$.

\item \label{RemReducC1} If $q \neq 1$ and $X_q$ is a suffix of $f(x_q)$,
i.e., $X_q'=\varepsilon$ (see Figure~\ref{remcas4}),
then we do not need $x_1$ and optionally not $w_1$ in the hypotheses of Lemma~\ref{Reduc1}.
Conclusion remains true with $u=X_1Y_1$, $w_2'=w_1w_2$ or $w_2$, $f(w_2')=pX_1Y_1X_2$, 
$w=w_2'x_2w_3...w_{\kappa}x_{\kappa}w_{{\kappa}+1}$, and $\check{w}$ a (not necessarily proper) suffix of $w_2'w_3...w_{\kappa}w_{{\kappa}+1}$

\item \label{RemReducC2} By mirror image of Case~\ref{RemReducC1}, we get that, 
if $q \neq {\kappa}$ and $Y_q$ is a prefix of $f(x_q)$ then we do not need $x_{\kappa}$ 
and optionally not $w_{{\kappa}+1}$ in the hypotheses of Lemma~\ref{Reduc1}.
Conclusion remains true with $u=X_{\kappa}Y_{\kappa}$, $w_{\kappa}'=w_{\kappa}w_{{\kappa}+1}$ or $w_{\kappa}$, $f(w_{\kappa}')=Y_{{\kappa}-1}X_{\kappa}Y_{\kappa}s$,
$w=w_1x_1w_2...w_{{\kappa}-1}x_{{\kappa}-1}w_{\kappa}'$, and 
$\check{w}$ a (not necessarily proper) prefix of $w_1w_2...w_{{\kappa}-1}w_{\kappa}'$.


\end{enumerate}
\end{rqe}

For any positive integer $\ell$, since $|f(x_i)|=|f(x_j)|$ 
is equivalent to $|f(x_i^\ell)|=|f(x_j^\ell)|$ and since a prefix (resp. a suffix) of $f(x_i)$
is a prefix (resp. a suffix) of $f(x_i^\ell)$, we immediately obtain the following Corollary 
that will be the central point of proof of Proposition~\ref{mainresult}.


\begin{cor} \label{Reduc2} (method of reduction){\rm \cite{Wla2015a}}

Let ${\kappa} \geq 3$ and $\ell \geq 1$ be two integers,
let $\alpha$ be an integer in $\{1,2\}$
and let $\beta$ be an integer in $\{{\kappa}-1,{\kappa}\}$

Let $f$ be a morphism \fromAtoB \/ and
let $(w_i)_{i=\alpha..\beta+1}$, $(x_i)_{i=\alpha..\beta}$ be words over $A$
such that $|f(x_i)|=|f(x_j)| \neq 0$ for all integers $i,j$ in $[\alpha,\beta]$.

We denote by $w$ the word $w_{\alpha} x_{\alpha}^\ell...w_{\beta}x_{\beta}^\ell w_{\beta+1}$.

We assume that there exist $u$, $p$, $s$, $(X_i)_{i=\alpha..\beta}$ and $(Y_i)_{i=\alpha..\beta}$ 
words over $B$ such that
$f(w_i)=Y_{i-1}X_{i}$ for all integers $i$ in $[1+\alpha;\beta]$.
Furthermore, we also assume that
$f(w_\alpha)= pu^{\alpha-1}X_1$ and $f(w_{\beta+1})=Y_{{\kappa}}u^{{\kappa}-\beta}s$ 
where
$u=X_if(x_i^\ell)Y_i (\neq \varepsilon)$
for all integers $i$ in $[\alpha,\beta]$: it means that $f(w)=pu^{\kappa}s$.
 
Finally, we assume that  there exists an integer $q$ such that, for any integer $i$ in $[\alpha,\beta]$, 
$0 \leq |X_q|- |X_i| \leq |X''_q|$ where $X''_q$ is a common suffix of $X_q$ and $f(x_q)$,
$0 \leq |X_q|- |X_i| \leq |f(x_q)|$ when $\alpha=2$,
or $0 \leq |Y_i|- |Y_q| \leq |f(x_q)|$ when $\beta={\kappa}-1$.

Then, for any integer $0 \leq \phi < \ell$, the word $\Modific{w}=w_{\alpha} x_{\alpha}^\phi...w_{\beta}x_{\beta}^\phi w_{\beta+1}$
satisfies $f(\Modific{w})=p\Modific{u}^{\kappa}s$ with
$\Modific{u}=X_if(x_i^\phi)Y_i$ for any integer $i$ in $[1;{\kappa}]$.

In particular, $f(\Modific{w})$ and $\Modific{u}^{\kappa}$ are synchronised only if $f(w)$ and $u^{\kappa}$ are synchronised.

\end{cor}

\section{Main result}

\begin{proposition}\label{mainresult}

For any integer $k \geq 3$,
an overlap-free morphism is a $k$-power-free morphism.

\end{proposition}

\begin{proof}

Let $f$ be a morphism from $A^*$ to $B^*$. We assume that $f$ is not $k$-power-free
and we want to show that $f$ is not overlap-free. 

The morphism $f$ must be a strongly bifix morphism. Otherwise, $f$ is not overlap-free and it ends the proof.

Let $w$ be a shortest $k$-power-free word whose image by $f$ contains a $k$-power.
Hence, $f(w)=p u^{k} s$ for two words $p$ and $s$ and a non-empty word $u$ in $B^+$.

If $f(w)$ and $u^{k}$ are synchronised then, by Lemma~\ref{synckp}, $w$ contains a $k$-power:
a contradiction.

Now, let us assume that $f$ is a strongly bifix morphism, and that $f(w)$ and $u^{k}$ are not synchronised.
In particular, it implies that $f$ is ps-morphism and injective. 

The central point of this proof is that, starting with $w$ and $u$, we use iteratively 
reduction of Lemma~\ref{Reduc2} (that is, of Lemma~\ref{Reduc1} and including the special cases of
Remark~\ref{RemReduc}) on the word whose image 
contains a $k$-power in such a way that there is no reduction left. 
That is, no situation of the hypotheses of Lemma~\ref{Reduc2} 
can be founded after this procedure. We obtain new words 
$W$ and $U$ such that $f(W)=PU^{k}S$
with $P$ a proper prefix of $\sub{W}{1}$, 
$S$ a proper suffix of $\sub{W}{|W|}$
and $f(W)$ and $U^{k}$ are not synchronised.

Moreover, if $W \neq w$, i.e., $|W|<|w|$, then, by definition of $w$, it means that $W$ contains a
$k$-power.

We will show that either $f$ is not overlap-free, or $f(W)$ and 
$U$ can again be reduced using Lemma~\ref{Reduc2}: a contradiction.

Another way to define $W$ is to choose a word of minimal length whose image by $f$ contains
a non-synchronised $k$-power. If we can reduce $W$ then this contradicts the hypothesis of minimal
length of $W$.

We focus on the fact that $W$ necessarily contains an overlap. Indeed, 
any factor of $U^k$ of length greater than $2|U|$ contains an overlap.
So, the contrary ends the proof: $f$ is not overlap-free.
More precisely, whatever the proper conjugated word
$U_c$ of $U$, $f(W)$ contains
$U_cU_cx$ where $x$ is the first letter of $U_c$.

\setcounter{comptnivun}{1}
{\textbf{\textit{Step \thecomptnivun}}}: 
For any pure overlap $avava$ of ${W}$ with $a \in A$
and $v\in A^*$, the words ${U}^{k}$ and $f(avava)$ do not have any common factor
of length at least $|{U}|+|f(av)|$.

By contradiction, let us assume that ${U}^{k}$ and $f(avava)$ have a common 
factor of length at least $|{U}|+|f(av)|$.
By Corollary~\ref{Kera},  
there exist two words $t_1$ and $t_2$,
and two integers $r$ and $q$ such that $f(av)=(t_1t_2)^r$ and
${U}=(t_2t_1)^q$
with $t_1t_2$ and $t_2t_1$ primitive words.

If $r \geq 2$ then $f(avav)=(t_1t_2)^{2 \times r}$ with $2 \times r \geq 4$. 
Therefore, $f(avav)$ contains an overlap with $avav$ overlap-free:~
$f$ is not overlap-free.

If $r=1$ then let $v_1$ and $v_2$ be the words such that $W=v_1\, avav \, v_2$. 
Since $f(W)=p(t_2t_1)^{q \times k} s$,
by Remark~\ref{RemPrimFactInt}, there exist two integers $\ell_1$
and $\ell_2$ such that $f(v_1)=p(t_2t_1)^{\ell_1}t_2$ and $f(v_2)=t_1(t_2t_1)^{\ell_2}s$ with
$\ell_1 + \ell_2 \geq q \times k-3$.

Since $f$ is strongly bifix, using Lemma~\ref{1pref} or Lemma~\ref{1suff}, by induction, we get that 
$W$ contains $(av)^{q}$:~a contradiction with the hypothesis that
$f(W)$ and $U^{k}$ are not synchronised.

\addtocounter{comptnivun}{1}
{\textbf{\textit{Step \thecomptnivun}}}: 
$\sub{{W}}{2..|{W}|-1}$ contains an overlap and so a pure overlap.

By contradiction, let us assume that $\sub{{W}}{2..|{W}|-1}$ is overlap-free.
In particular, it implies that $W$ starts or ends with a pure overlap.
Let $s_1$ and $p_{k+1}$ be the words such that
 $f(\sub{{W}}{1})=ps_1$ and $f(\sub{{W}}{|{W}|})=p_{k+1}s$,
that is, ${U}^{k}=s_1f(\sub{{W}}{2..|{W}|-1})p_{k+1}$.

If $|s_1|+|p_{k+1}| < |U^{k-2}|$, then 
$|f(\sub{{W}}{2..|{W}|-1})| > 2|U|$ and so contains an overlap:~$f$ is not overlap-free.

If $|s_1|+|p_{k+1}| \geq |U^{k-2}| \,  (\geq |U|)$ then either $|s_1| \geq |U|/2$
or $|p_{k+1}| \geq |U|/2$.
If $|s_1| \geq |U|/2$ and $W$ starts with a pure overlap $avava$ with $a\in A$ and $v \in A^*$.
It implies that $s_1$ is a suffix of $f(a)$.
If $W=avava$ then $|s_1f(vav)p_{k+1}|\geq |U|+|f(va)|$ with $s_1f(vav)p_{k+1}$
a common factor of a power of $U$ and a power of $f(av)$:~a contradiction with Step 1.
If $W \neq avava$ then $|s_1f(vava)|\geq 2|s_1|+|f(va)| \geq |U|+|f(va)|$ with $s_1f(vava)$
a common factor of a power of $U$ and a power of $f(av)$:~a contradiction with Step 1.
In a same way, If $|p_{k+1}| \geq |U|/2$, $W$ can not end with a pure overlap.

If $|s_1| < |U|/2$ then $|p_{k+1}| > |U|/2$ and $W$ do not end with a pure overlap.
So either $\sub{{W}}{2..|{W}|-1}$ contains an overlap or $\sub{{W}}{2..|{W}|}$ is overlap-free. 
But, in this second case, $f(\sub{{W}}{2..|{W}|-1})p_{k+1}$ is factor of $U^k$ with 
$|f(\sub{{W}}{2..|{W}|-1})p_{k+1}|=k|U|-|s_1| >2|U|$.
So $f(\sub{{W}}{2..|{W}|})$ contains an overlap:~$f$ is not overlap-free.

In the same way, if $|p_{k+1}| < |U|/2$, we obtain that $\sub{{W}}{2..|{W}|-1}$ contains an overlap or 
that $f$ is not overlap-free.

\vsd

Let $POI$ be the set of pure overlaps that are factors of $\left(\sub{{W}}{2..|{W}|-1} \right)$.

The set $\left\{|f(avava)| \mid avava \in POI \right\}$ is finite and admits a minimum $L_{min}$.

Let $POM = \left\{  avava \in  POI  \mid |f(avava)| =L_{min} \right\}$.

Let us also recall that, if $avava$ is a  pure overlap, then $a\in A$.

\addtocounter{comptnivun}{1}
{\textbf{\textit{Step \thecomptnivun}}}: 
For any overlap $avava \in POM$, the word $f(avava)$ is an internal factor of ${U}^2$.

For any pure overlap $avava \in POI$, since $f(avava)$ is an internal factor of $U^k$, 
by Step~1, we have $|f(avava)| < |U|+|f(av)|$.
Since $|f(avava)| \geq 2|f(a)|$, we can not have $|f(avava)| \geq 2|U|$.

Let $W_1$ and $W_2$ be the non-empty words such that $W=W_1avavaW_2$.
By contradiction, let us assume that $f(avava)$ (with $|f(avava)| < 2|U|$)
is not an internal factor of ${U}^2$.
It means that there exist two non-empty words $v_1$ and $v_2$ such that
$f(avava)=v_1Uv_2$.
In particular, $v_1$ is a suffix of $U$, $v_2$ is a prefix of $U$ and $|v_1v_2|<|U|$.
Since $|f(ava)|<|{U}|$, we get that $|v_1v_2|=|f(av)|+|f(ava)|-|U|<|f(av)|$.

\setcounter{comptnivdeux}{1}
{\textit{Case \thecomptnivun.\thecomptnivdeux:}} $|v_1| \leq |f(a)|$ and $|v_2| \leq |f(a)|$

Let $A_1$ and $A_2$ be the words
such that $v_1$
such that $f(a)=v_1A_1=A_2v_2$ and $U=A_1f(vav)A_2$.

If $|v_1v_2| = |f(a)| $, i.e.,  $|v_1|=|A_2|$ then $|f(avav)|=|U|$:~a contradiction with the assumption that
$f(W)$ and $U^k$ are not synchronised.

In a same way, if $|v_1| = |f(a)|$ and $|v_2| = |f(a)|$ then $|f(vav)|=|U|$:~again a contradiction with the
assumption that $f(W)$ and $U^k$ are not synchronised.
It means that $A_1A_2 \neq \varepsilon$.

If $|v_1|>|A_2|$ then $|v_2|>|A_1|$. Since $v_1$ and $A_2$ are both suffixes of $U$,
let $X$ be non-empty the word such that $v_1=XA_2$.
Since $v_2$ and $A_1$ are both prefixes of $U$, let $Y$ be the non-empty word such
that $v_2=A_1Y$.
We get $Xf(a)=f(a)Y$ with $|X|=|Y|=|v_1|-|A_2|=|v_2|-|A_1| < |f(a)|$. 
By Case~\ref{Lotcase1} of Proposition~\ref{Lothaire}
and Remark~\ref{RemLotC1}, we get that $f(a)Y$, prefix of $f(ava)$, contains an overlap:~$f$ is
not overlap-free.

If $|v_1|<|A_2|$ then $|v_2|<|A_1|$. Let $X$ be the non-empty word such that $Xv_1=A_2$.
Since $v_2$ and $A_1$ are both prefixes of $U$, let $Y$ be the word such
that $A_1=v_2Y$.
We get $Xf(a)=f(a)Y$ with $|X|<|f(a)|$. By Case~\ref{Lotcase1} of Proposition~\ref{Lothaire},
there exist two words $r$ and
$s$ in $A^*$, and an integer $n$ such that $f(a)=r(sr)^{n}$, $X=rs$
and $Y=sr$. Since $|X|<|f(a)|$, we get that $n \geq 1$ when $r \neq \varepsilon$ and 
$n \geq 2$ when $r = \varepsilon$. But, if $s=\varepsilon$ or if $n \geq 2$, 
since $f(a)$ or $f(aa)$ contains an overlap, 
we get that $f$ is not overlap-free:~it ends the proof. So we may assume that $n=1$,
$r \neq \varepsilon$ and $s \neq \varepsilon$.

Let $W_1'$ the smallest suffix of $W_1$ such that $f(W_1')$ ends with $X$.
There exist a word $p_1'$ such that  $f(W_1')=p_1'X$ with $|p_1'|<|f(\sub{W_1'}{1})|$. 
It follows that  $f(W_1'a)=p_1'rsrsr$.
If $W_1'a$ is overlap-free, it ends the proof:~$f$ is not overlap-free.
Since $|f(W_1')|=|p_1'sr|$ and $|f(a)|=|rsr|$, if $W_1'a$ contains an overlap $btbtb$ then $b \neq a$
and $|f(btbtb)|<|f(avava)|$:~a contradiction the hypothesis of the minimal length of 
$f(avava)$.

\addtocounter{comptnivdeux}{1}
{\textit{Case \thecomptnivun.\thecomptnivdeux:}} $|v_1| > |f(a)|$


Since $f(a)$ is the suffix of $Uv_2$ and of $v_1v_2$ of length $|f(a)|$,
let $X$ be the word such that $v_1v_2=Xf(a)$.
And, since $f(a)$ is a prefix of $v_1$, let $Y$ be the word such that $v_1v_2=f(a)Y$.

If $|Y|=|f(a)|$ then $|f(vava)|=|U|$:~a contradiction with the assumption that
$f(W)$ and $U^k$ are not synchronised.

If $|Y|<|f(a)|$, by Case~\ref{Lotcase1} of Proposition~\ref{Lothaire}
and Remark~\ref{RemLotC1}, we get that $f(a)Y$, prefix of $f(ava)$, contains an overlap:~$f$ is
not overlap-free.

If $|Y|>|f(a)|$, let $Z$ be the non-empty word such that $Y=Zf(a)$ and $X=f(a)Z$.
Since $|f(av)|>|v_1v_2|$, we get that  $v \neq \varepsilon$, $Y$ is a prefix of $f(v)$ and 
$X$ is a suffix of $f(v)$. It follows that $f(vav)$ contains the overlap $f(a)Zf(a)Zf(a)$:~$f$ is not overlap-free.

\addtocounter{comptnivdeux}{1}
{\textit{Case \thecomptnivun.\thecomptnivdeux:}} $|v_2| > |f(a)|$

It is the mirror image of Case \thecomptnivun.2.

\addtocounter{comptnivun}{1}
{\textbf{\textit{Step \thecomptnivun}}}: Decomposition of $U^k$.

For every integer $j$ in $[1,k+1]$, 
let $i_j$ be the smallest integer such that $PU^{j-1}$ is a prefix
of $f(\sub{W}{1..i_j})$.
We have $i_1=1$ and 
there exist some words $p_j$ and $s_j$ such that
$f(\sub{W}{i_j})=p_js_j$, $p_1=P$,
$s_{k+1}=S$,
$p_j \neq \varepsilon$ if $j \neq 1$,
and $s_1 \neq \varepsilon$.

Let $avava$ be an overlap and let us assume that $avava \in POM$.
If $i_j=i_\ell$ with $j < \ell$  then, by definition of $i_j$ and $i_\ell$, it means
that $f(\sub{W}{i_j})$ contains $U^{\ell-j}$.
If $f(avava)$ is factor of $U$ that is of $f(\sub{W}{i_j})$ then $f$ is not overlap-free.
If $f(avava)$ is factor of $U^2$ (but not of $U$) then 
we must have $j+1=\ell$ (and $i_{j+1}\neq i_{j+2}$) otherwise
$f$ is not overlap-free (and it ends the proof).
It follows that $f(avava)$ is a factor of $f(\sub{W}{i_j..i_{j+2}})$.
Let $\alpha \geq i_{j}$ be the greatest integer and 
$\beta \leq i_{j+2}$ be the lowest integer such that $f(avava)$ is a factor of $f(\sub{W}{\alpha..\beta})$
that is $f(\sub{W}{\alpha..\beta})=\pi_1f(auaua)\sigma_2$
with $|\pi_1|<|f(\sub{W}{\alpha})|$ and $|\sigma_2|<|f(\sub{W}{\beta})|$. 
By Lemma~\ref{LemConjImOver}, either $f$ is not overlap-free (and it ends the proof) or
$\sub{W}{\alpha..\beta}$ starts with an overlap. 
In this second case,  since $f(avava)$ is not factor of $U$, we get that $\alpha=i_j$.
It follows that $\sub{W}{i_j..\beta}$ contains at least three occurences of $\sub{W}{i_j}$
that is $|f(\sub{W}{i_j..\beta})| \geq 3|U|$: a contradiction with $i_j=i_{j+1}$ and $\beta \leq i_{j+2}$.

So we get that $i_j \neq i_\ell$ when $j \neq \ell$.
It follows $f(\sub{\omega}{1..i_{j}})=pU^{j-1}s_{j}$ 
for every integer $j$ in $[1,k+1]$ and
$U=s_jf(\sub{\omega}{i_j+1..i_{j+1}-1})p_{j+1}$
for every integer $j$ in $[1,k]$.

\addtocounter{comptnivun}{1}
{\textbf{\textit{Step \thecomptnivun}}}: Case where there exists an overlap $avava \in POM$
such that $f(avava)$ is an internal factor of $U$.

It means that, for every integer $j$ in $[1,k]$, 
$f(avava)$ is an internal factor of $f(\sub{\omega}{i_j..i_{j+1}})$.
Let $T_j$ be the shortest factor of $f(\sub{\omega}{i_j..i_{j+1}})$ that contains $f(avava)$.

By Lemma~\ref{LemConjImOver} and Remark~\ref{RemConjImOver}, 
there exist a letter $b_j$ and a word $v_j$ such that $T_j$ starts with $b_jv_jb_jv_jb_j$
with $f(v_jb_j)$  an internal factor
of $f(au)f(au)$ and $|f(b_jv_j)|=|f(v_jb_j)|=|f(au)|$.
By Lemma~\ref{Reduc1} and Remark~\ref{RemReduc1}, a reduction can be done:~a contradiction
with the hypotheses.

\addtocounter{comptnivun}{1}
{\textbf{\textit{Step \thecomptnivun}}}: Case where, for any overlap  $avava \in POM$,
$f(avava)$ is an internal factor of ${U}^2$ but not of $U$.

For any $avava \in POM$, let $V_{1,av}$ and 
$V_{2,av}$ be the words such that $f(avava)=V_{1,av}V_{2,av}$
with $V_{1,av}$ a suffix of $U$ and $V_{2,av}$ a prefix of $U$.
In particular, we have either $|f(av)|<|V_{1,av}|$ or $|f(av)|<|V_{2,av}|$.
For any integer  $j$ in $[1,k-1]$, 
$f(avava)$ is an internal factor of $f(\sub{\omega}{i_j..i_{j+2}})$.
Let $T_j$ be the shortest factor of $f(\sub{\omega}{i_j..i_{j+2}})$ that contains $f(avava)$.

Let $\alpha \nu \alpha \nu \alpha \in POM$ such that 
$|V_{1,\alpha \nu}|= \max \{|V_{1,av}| \mid avava \in POM\}$ 
and 
let $\beta \mu \beta \mu \beta \in POM$ such that 
$|V_{2,\beta \mu}|= \max \{|V_{2,av}| \mid avava \in POM\}$.

Let $j_\alpha$ be an integer such that 
$\alpha \nu \alpha \nu \alpha \in \factors{\sub{\omega}{i_{j_\alpha}..i_{j_\alpha+2}}}$.

For any integer  $j$ in $[1,k-1]$, 
let $T_j$ be the shortest factor of $f(\sub{\omega}{i_j..i_{j+2}})$ that contains $f(\alpha \nu \alpha \nu \alpha)$.
Since $f(W)$ and $U^k$ are not synchronised, there exists an integer $r \in \{j_\alpha-1,j_\alpha+1\}$ 
such that $T_r \neq \alpha \nu \alpha \nu \alpha$.
By Lemma~\ref{LemConjImOver} and Remark~\ref{RemConjImOver}, 
there exist a letter $b_r$ and a word $v_r$ such that $T_r$ starts with $b_rv_rb_rv_rb_r$ with
$|f(b_rv_r)|=|f(\alpha \nu)|$.
Since $|V_{1,b_r v_r}|<|V_{1,\alpha \nu}|$, it implies that $b_r =\sub{W}{i_r}$ and $V_{1,b_r v_r}=p_r$.
Since $f(b_rv_rb_rv_rb_r)$ is an internal factor of $U^2$,
it follows that $|V_{1,\alpha \nu}| > 2|f(v_r b_r )| = 2|f(\alpha \nu)|$.

Moreover, we have $|V_{2,\beta \mu}| \geq |V_{2,b_r v_r}| = |s_r| + 2|f(v_r b_r )|   \geq 2|f(\beta \mu)|$.





By Lemma~\ref{LemConjImOver} and Remark~\ref{RemConjImOver}, 
for any integer  $j \in [1,k-1]$, there exist a letter $b_j$ and a word $v_j$ such that 
$T_j$ starts with $b_jv_jb_jv_jb_j$
with $f(b_jv_j)$  an internal factor
of $f(\beta \mu)f(\beta \mu)$ and $|f(b_jv_j)|=|f(\beta \mu)|$.

Since $|V_{2,\beta \mu}| \geq 2|f(\beta \mu)|$, for any integer  $j \in [1,k-1]$, there exist a word 
$X_j$ such that $V_{2,b_j v_j}$ starts with $f(v_jb_j)$ or with $f(b_jv_j)$ and $|X_j|<|f(b_jv_j)|=|f(v_jb_j)|$.
Let $q$ be the integer such that $|X_q|=\max \{|X_i| \mid j \in [1,k-1]\}$.
By Lemma~\ref{Reduc1}, Point~\ref{RemReducC1} of Remark~\ref{RemReduc} and Remark~\ref{RemReduc1}, 
a reduction can be done:~a contradiction with the hypotheses.
\end{proof}

\bibliography{bf.bib}

\end{document}